\documentclass[10pt,final,journal, letterpaper,  twoside, twocolumn]{IEEEtran}
\usepackage{cite,graphics,amsmath,rotating,amssymb}

\newtheorem{theorem}{\text{Theorem}}
\newtheorem{corollary}{\text{Corollary}}

\newtheorem{proposition}{\text{Proposition}}
\newtheorem{remark}{\text{Remark}}
\allowdisplaybreaks

\title{Guiding Blind Transmitters: Degrees of Freedom Optimal
  Interference Alignment Using Relays}

\author{

  Ye~Tian,~\IEEEmembership{Student Member,~IEEE,} and
  Aylin~Yener,~\IEEEmembership{Member,~IEEE}

  \thanks{ Manuscript received May 31, 2012; revised Feb 6, 2013; accepted Apr 1, 2013. This work was supported in part by NSF
    grants 0721445, 0964364 and 0964362. This work was presented in
    part at the IEEE Wireless Communications and Networking
    Conference, WCNC'12, IEEE International Workshop on Signal
    Processing Advances in Wireless Communications, SPAWC'12 and IEEE
    International Symposium on Information Theory, ISIT'12.}

  \thanks{
      Ye Tian was with the Department of Electrical Engineering at the Pennsylvania State University, University Park, PA 16802; He is now with Broadcom Corporation, Sunnyvale, CA, 94085.
      Aylin Yener is with the Department of Electrical Engineering at
      the Pennsylvania State University, University Park, PA 16802
      (email: yetian@broadcom.com, yener@ee.psu.edu).

    }
    \IEEEpubid{}
}

\begin{document}
\thispagestyle{headings}
\maketitle

\begin{abstract}
Channel state information (CSI) at the transmitters (CSIT) is of importance
for interference alignment schemes to achieve the optimal degrees of
freedom (DoF) for wireless networks. This paper investigates the
impact of half-duplex relays on the degrees of freedom (DoF) of the X
channel and the interference channel when
the transmitters are blind in the sense that no CSIT is available. In
particular, it is shown
that adding relay nodes with global CSI to the communication model is
sufficient to recover the DoF that is the
optimal for these models with global CSI at the transmitters. The
relay nodes in essence help steer the directions of the transmitted signals to
facilitate interference alignment to achieve the optimal DoF with CSIT. The general $M\times N$ X
channel with relays and the $K$-user interference channel are both
investigated, and sufficient conditions on the number of antennas at
the relays and the number of relays needed to
achieve the optimal DoF with CSIT are established. Using relays,
the optimal DoF can be achieved in finite channel uses. The DoF for
the case when relays only have delayed CSI is
also investigated, and it is shown that with delayed CSI at the relay
the optimal DoF with full CSIT cannot be achieved. Special cases of the X channel and
interference channel are investigated to obtain further design
insights. 
\end{abstract}


\section{Introduction}
Interference is inherent to any fully connected multi-user wireless network. As the number of
devices sharing the spectrum with high rate demands grows, wireless
networks become more and more interference limited. The significance
of interference on the operation of a wireless network renders it
natural to focus on its high SNR performance to obtain
design insights and characterize the interaction between the signals.
Thus, degrees of freedom (DoF), which characterizes the scaling of the
transmission rates of
wireless networks in high signal to noise ratio (SNR) regime, is an
important metric to measure the performance of an interference-limited
system.

Interference alignment was shown to achieve
the optimal DoF for a variety of interference-limited wireless
networks \cite{DoFMIMOX,CadambeXNetwork,CadambeKuser,HostMadsen}. In reference \cite{DoFMIMOX}, the authors have shown that
the optimal DoF
$\frac{4M}{3}$ can be achieved for the 2-user multiple input multiple output
(MIMO) X channel with $M$ antennas at
each node, using symbol extensions and interference alignment,
demonstrating the achievability of non-integer DoF $4M \over 3$ with
constant channel for $M>1$. For $M=1$ with constant channels, the DoF
$\frac{4}{3}$ is shown to be achievable in \cite{HostMadsen}. Reference
\cite{CadambeXNetwork} further
generalized the result to the $M\times N$ user X channel, and showed that
the optimal DoF is $\frac{MN}{M+N-1}$ with
single antenna nodes and a time-varying channel. Reference
\cite{CadambeKuser} showed that interference
alignment achieves the optimal DoF of the $K$-user interference
channel, $\frac{K}{2}$, with single antenna nodes and time-varying channel. Follow up studies on the
DoF of the interference channels, for example, the SIMO interference channel,
the $K$-user $M\times N$ MIMO interference channel, and interference
channel with cooperation and cognition, can be found
in references \cite{TiangaoKMN,SIMOIC,DoFMIMOCC}.

To effectively implement interference alignment, it is crucial
to have global instant CSIT which can be difficult to obtain
for practical systems. Reference
\cite{DoFwithoutCSIT} has studied the DoF region of the $2$-user MIMO
broadcast channel and the $2$-user MIMO interference channel without CSIT,
and loss of DoF is observed for many scenarios of interest. Reference
\cite{VazeNoCSIT} has further generalized the results to $K$-user
broadcast and interference channels,
and also derived outerbounds on DoF region for the $K$-user X channel.
This reference has established that without CSIT, the transmitters cannot
steer the signals to the exact desired directions to guarantee that
the interference is aligned together at the receivers, which causes
the performance degradation in terms of DoF.

While loss of DoF is observed when no CSIT is available at the
transmitters, reference \cite{JafarCorr} has observed that as long as
the channel's correlation structure is known at the transmitters, without any knowledge of the
exact channel coefficient, interference alignment is
still possible for certain wireless networks. Reference
\cite{ChenweiDark} has further developed this idea and proposed the blind
interference alignment strategies using staggered antennas, which
can artificially create the desired channel correlation pattern by
switching the antennas used by the receivers. For systems where CSIT is
completely unknown however, loss of DoF appears to be inevitable.

A more practical assumption about CSIT is that the transmitter may have
delayed CSI. The delayed
CSIT model characterizes the channel variation and the delay
in the feedback of CSI from receivers, and thus is important from both theoretical and practical perspective.
The delayed CSIT assumption is first studied in the context of the
$K$-user broadcast channel \cite{StaleCSIT}, i.e., a channel with a transmitter having $K$ antennas and
$K$ receivers each with a single antenna, where the transmitter has
accurate and global CSIT delayed by several time slots. It is shown that
the delayed CSIT can be useful for interference alignment and
the DoF can be improved significantly compared to the case without CSIT.
This delayed CSIT assumption is then applied to various channel
models such as the general broadcast channels, interference
channels and X channels, and improvement
on the DoF compared to the cases without CSIT can be found in references
\cite{MalekiRetro,VazeMIMOICdelayed,VazeBCdelayed,GhasemiMIMOICdelayed,GhasemiXdelayed,TandonDof,MohammadBC}.

The delayed CSIT is an interesting assumption which, in fact, shows
that feedback of delayed CSI can provide capacity gain for multi-destination wireless
networks, which is in contrast with various single-destination models \cite{Shannonfeedback,VisFeedback,Basherfeedback}.
However, there is a performance degradation with the delayed CSIT
assumption compared to when global CSIT is available. For example,
the DoF for the $K$-user broadcast channel with delayed CSIT is shown
in \cite{StaleCSIT} to be
$\frac{K}{1+\frac{1}{2}+\cdots+\frac{1}{K}}$, whereas with global CSIT,
the optimal DoF is $K$.

The operation of relaying, although is beneficial in
improving the achievable rates for many multi-user wireless networks
\cite{cover_relay,Offset,tandonMARC,Myjournal1,YEICOBR,Sahin_ICOBR,ivana3},
is shown in reference \cite{CadambeALLDOF} to be unable to provide DoF
gain for the fully connected interference channel and X channel with full
CSI at all nodes. In this context, relaying is shown to be useful only
to facilitate
interference alignment for some specific scenarios. For quasi-static
channels, references
\cite{NouraniX,NouraniIC,IAXrelay} have proposed strategies to utilize the
relay to randomize the channel coefficients at the receivers, and the optimal DoF can be achieved although the channel is not time
varying within the transmission blocks. Reference \cite{Haishi} has proposed relay-aided
interference alignment schemes that can achieve the optimal DoF of the $K$-user
interference channel with finite time extensions. For networks that have limited CSI, it is shown
in reference \cite{Tannious} that using a relay, the
optimal DoF for the $K$-user interference channel can be achieved when all the nodes
have local CSI only, provided that the relay has more
antennas than the total number of single-antenna
transmitters. For networks that
are not fully connected, for example, the multi-hop relay networks, references
\cite{RankovRelay,JeonRelayNetwork,IFneu,TandonDof} have studied the
DoF under either global CSIT or delayed CSIT assumptions.

Whereas the study of relaying on the DoF of fully connected wireless networks
so far focused on using relays to facilitate interference alignment,
in this work, we aim to theoretically study the
impact of relaying on the DoF from another perspective. We focus on
understanding whether, and to what extent, relays can improve the DoF of wireless
networks when the source nodes, i.e., the transmitters, are {\it blind} in the sense that no
CSIT is available. In this paper, we mainly consider the case when
relays have global CSI as a first step to investigate the impact of
relays on the DoF of wireless networks without CSIT. The
justification of the setting is that it is likely that the relay nodes
are located in between the sources and the destinations and could have
access to more accurate CSI. The relays can be small base
stations at fixed locations with more power resources and computing
capability, and obtaining CSI can be less challenging. Specifically,
we study the DoF of the X channel and the
interference channel with single-antenna users and half-duplex
multi-antenna relays, where no CSI is
available at the transmitters, but global CSI can be obtained at the relays and the
receivers. We first design a {\it joint beamforming} based transmission scheme for the general $M\times N$
X channel with relays. We show that when each relay is equipped with $L$ antennas, with
$\left\lceil\frac{(M-1)(N-1)}{L^2}\right\rceil$ relays, the DoF $\frac{MN}{M+N-1}$
can be achieved, which is the same optimal DoF as the case when CSIT is available. We
then consider two special cases: the $K$-user X
channel with a multiple antenna relay and the $K$-user X channel with
single antenna relays. For the case when
relay nodes are only equipped with single antenna, the optimal DoF can be achieved with $(K-1)^2$
relays. For the case with one multiple
antenna relay, we can design a different
scheme with less computational complexity that uses {\it partial interference alignment} at the relay and
{joint beamforming} to show the achievability of optimal DoF
$\frac{K^2}{2K-1}$ using one relay with $K-1$ antennas. Note that in the above
results, the channel is required to be time varying in order to
achieve the optimal DoF. An
interesting feature of the DoF optimal interference alignment scheme
using relays is that only finite channel usage is required to achieve
the exact optimal DoF, whereas for the general $M\times N$ X channel without relays but with CSIT,
infinite channel uses are required. The case when there is no CSIT but
relays only have delayed CSI is also investigated.

Using the techniques developed for the X channel, we further show that
interference alignment is possible for the $K$-user interference
channel without CSIT with the help of half-duplex relays. For the
general case, we design a two-slot transmission scheme using joint
beamforming, and show that it requires $\left\lceil\frac{K(K-2)}{L^2}\right\rceil$
relays with $L$ antennas to achieve the DoF $K \over 2$, which is
exactly the same optimal DoF as the case with CSIT. We then consider
two special cases: the case with one relay
with $K-1$ antennas, and the case with $K(K-2)$ single antenna
relays. Note that the special case when the relay has $K-1$ antennas
is also investigated in reference \cite{Chen_Globecom}. When we
have one relay with $K-1$ antennas, joint beamforming is not necessary
for interference alignment and the channel does not need to be time varying. However, when we have $K(K-2)$
relays each with a single antenna, joint beamforming is required to achieve
interference alignment and the channel does need to be time varying. 

Throughout the paper, we use bold letters, e.g. ${\bf h}$, to denote
constant vectors, bold capital letters, e.g. ${\bf H}$, to
denote matrices or vector of random variables, and ordinary capital letters, e.g. $H$, to denote
random variables. We use $\lceil x \rceil$ to denote the closest integer
that is smaller than $x$, and $\lfloor x \rfloor$ to denote the
closest integer that is larger than $x$. $[a_i]^{(i)}$ denotes the column vector obtained by
enumerating $a_i$ with index $i$, i.e.,
\begin{equation}
[a_i]^{(i)}=[a_i]_{i=1}^{i=n}=\left[a_1,a_2,\cdots,a_n\right]^T,\label{eq:8}
\end{equation} if $i=1,2,\cdots,n$.

Similarly, $[a_{ij}]^{(ij)}$ denotes the column vector
\begin{equation}
\left[\left[\left[a_{1j}\right]^{(j)}\right]^T,\left[\left[a_{2j}\right]^{(j)}\right]^T,\cdots\right]^T,\label{eq:11}
\end{equation} which is obtained by enumerating $a_{ij}$ for all
indices $i$ and $j$ as its entries.

We also use $[a_{ijk}]^{(ijk)}$ to denote the column vector
\begin{equation}
\left[\left[\left[a_{1jk}\right]^{(jk)}\right]^T,\left[\left[a_{2jk}\right]^{(jk)}\right]^T,\cdots\right]^T\label{eq:10},
\end{equation} which is obtained by enumerating $a_{ijk}$ for all
indices $i$, $j$, $k$ as its entries.

The remainder of the paper is organized as follows: Section
\ref{sec:system-model} introduces the system model. Section
\ref{sec:relay-aided-interf-XX} studies the relay-aided interference
alignment schemes for the X channel. Section
\ref{sec:relay-aided-interf-IC} studies the relay-aided interference
alignment schemes for the interference channel. Section
\ref{sec:conclusion} concludes the paper.

\section{System Model}\label{sec:system-model}
\subsection{$M\times N$ X channel with Relays}\label{sec:mtimes-n-x}
\begin{figure}[t]
\centering
\includegraphics[width=3in]{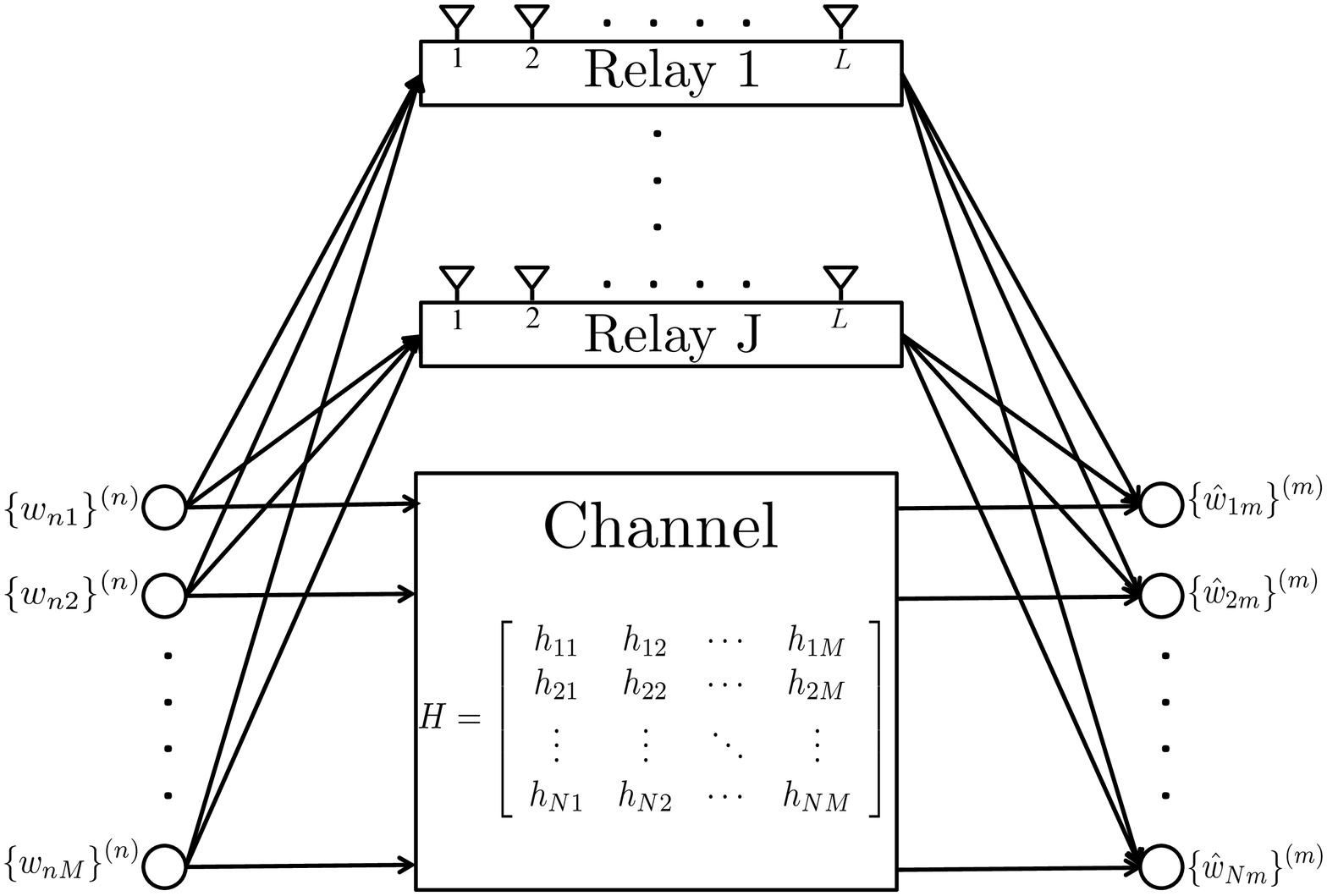}
\caption{$M\times N$ X channel with relays with $m=1,\cdots,M$, $n=1,\cdots,N$.} \label{fig-Xrelay}
\end{figure}
Fig. \ref{fig-Xrelay} shows the $M\times N$ X channel with relays.
In this model, there are $M$ transmitters and $N$ receivers, and each transmitter
has a message to be communicated to each receiver. It is assumed that
the transmitters and receivers are equipped with single
antenna. There are $J$ half-duplex relays available to help the
transmission. Each relay is assumed to have $L$ antennas. We denote $w_{nm}$ as the message
intended from transmitter $m$ to receiver $n$. The transmitted signal
from transmitter $m$ is denoted as $X_m(t)\in \mathbb{C}$ and the
transmitted signal from relay
$R_j$ is denoted as ${\bf X}_{R_j}(t) \in \mathbb{C}^L$, where $t$ is the time
index denoting the slot in which the signal is transmitted.

When the relays listen to the channel, the received signals at the
receivers are
\begin{equation}
Y_n(t)=\sum_{m=1}^Mh_{nm}(t)X_m(t)+{Z}_{n}(t),\label{eq:12}
\end{equation} where $~n=1,\cdots
N,~m=1,\cdots,M$, and the received signals at the relays are
\begin{equation}
{\bf Y}_{R_j}(t)=\sum_{m=1}^M{\bf h}_{R_jm}(t)X_m(t)+{\bf
  Z}_{R_j}(t),~j=1,\cdots,J.\label{eq:13}
\end{equation}

When the relays transmit, the received signals at the receivers are
\begin{equation}
Y_n(t)=\sum_{m=1}^Mh_{nm}(t)X_m(t)+\sum_{j=1}^J{\bf h}_{nR_j}(t)^T{\bf X}_{R_j}(t)+{Z}_{n}(t).\label{eq:14}
\end{equation}

In the above expressions, the transmitted signals are subject to
average power constraints $E(||{\bf X}_{R_j}(t)||^2)\le P$,
$E(|{X}_{m}(t)|^2)\le P$, $j=1,\cdots,J$, $m=1,\cdots,M$. $h_{nm}\in \mathbb{C}$ is the channel
coefficient from transmitter $m$ to the receiver $n$. ${\bf h}_{R_jm}(t) \in
\mathbb{C}^{L}$ is the channel vector between transmitter $m$ and
relay $R_j$, and ${\bf h}_{nR_j}(t) \in
\mathbb{C}^{L}$ is the channel vector between relay $R_j$ and receiver
$n$. It is assumed that the channel coefficients are independently drawn from a
continuous distribution for each time index, and the channel is time
varying. $Z_n(t)$ and ${\rm
  \bf Z}_{R_j}(t)$
are zero-mean Gaussian random variables with unit
variance and identity covariance matrix, respectively.

We denote the rate of message $w_{ij}$ with $R_{ij}(P)$ under power
constraint $P$. Define $\mathcal{C}(P)$ as the set of all achievable rate tuples $[R_{nm}(P)]^{(nm)}$ under
 power constraint $P$. The DoF is defined as
\begin{equation}
DoF = \lim_{P\to \infty} \frac{R_{\sum}(P)}{\log (P)},\label{eq:15}
\end{equation}
where $R_{\sum}(P) = \max_{\mathcal{C}(P)}
\left(\sum_{m,n}R_{nm}(P)\right)$. Note that since we consider
the DoF as our metric, in the rest of the paper, we omit the
noise terms in equations \eqref{eq:12}-\eqref{eq:14}.

\subsection{$K$-user Interference Channel with Relays}
\begin{figure}[t]
\centering
\includegraphics[width=2.7in]{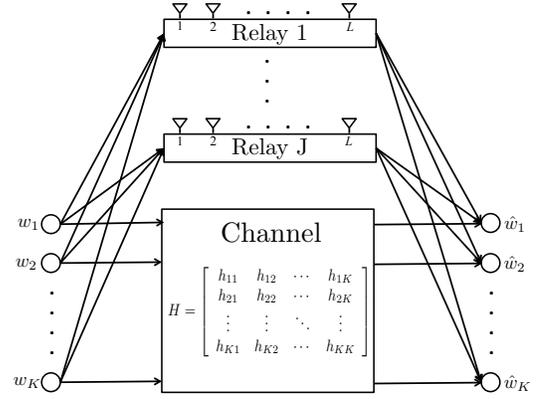}
\caption{$K$-user interference channel with relays.} \label{fig-ICrelay}
\end{figure}
Fig. \ref{fig-ICrelay} shows the $K$-user interference channel with relays.
In this model, there are $K$ transmitters and $K$ receivers, and each transmitter
has a message to be communicated to one intended receiver. It is assumed that
the transmitters and receivers are equipped with single
antenna. There are $J$ half-duplex relays available to help the
transmission. Each relay is assumed to have $L$ antennas. We denote $w_{k}$ as the message
intended from transmitter $k$ to receiver $k$, $k=1,\cdots,K$. The transmitted signal
from transmitter $k$ is denoted as $X_k(t)\in \mathbb{C}$ and the signal from relay
$j$ is denoted as ${\bf X}_{R_j}(t) \in \mathbb{C}^L$, where $t$ is the time
index denoting the slot in which the signal is transmitted.

When the relays listen to the channel, the received signals at the
receivers are
\begin{equation}
Y_n(t)=\sum_{k=1}^Kh_{nk}(t)X_k(t)+{Z}_{n}(t),~n=1,\cdots,K\label{eq:16}
\end{equation} and the received signals at the relays are
\begin{equation}
{\bf Y}_{R_j}(t)=\sum_{k=1}^K{\bf h}_{R_jk}(t)X_k(t)+{\bf Z}_{R_j}(t),
~J=1,\cdots,J.\label{eq:21}
\end{equation}

When relays transmit, the received signals at the receivers are
\begin{equation}
Y_n(t)=\sum_{k=1}^Kh_{nk}(t)X_k(t)+\sum_{j=1}^J{\bf h}_{nR_j}(t)^T{\bf X}_{R_j}(t)+{Z}_{n}(t).\label{eq:22}
\end{equation}

The power constraints on the transmitted signals, the channel
coefficients, and the channel noise are defined as in Section \ref{sec:mtimes-n-x}.

We denote the rate of message $w_{k}$ is $R_{k}(P)$ under power
constraint $P$. Define $\mathcal{C}(P)$ as the set of all achievable rate tuples $[R_{k}(P)]^{(k)}$ under
 power constraint $P$. The DoF is defined as in
 (\ref{eq:15}) with the sum rate now defined as $R_{\sum}(P) = \max_{\mathcal{C}(P)}
 \left(\sum_{k=1}^KR_{k}(P)\right)$. We ignore the noise terms in
 equations \eqref{eq:16}-\eqref{eq:22} in the sequel.

\section{Relay-aided Interference Alignment for X Channel without
  CSIT}\label{sec:relay-aided-interf-XX}
In this section, we provide the DoF for the $M\times N$ X channel, with the assumption that the
transmitters have no CSI, and relay nodes with global CSI are present to
help. Without CSIT, the transmitters cannot send the signals in the desired
directions to align the interference at the
receivers. However, as we shall show next, relays can be used to help the transmitters steer the
directions of the transmitted signals to achieve the DoF as if global CSI
were available at the transmitters. 

Before presenting the relay-aided interference alignment schemes, we
first find an upper bound for the DoF of the $M\times N$ X
channel without CSIT, but with relays.

\begin{proposition}
\label{prop-upperbound-forX} For the $M\times N$ X channel without
CSIT, where relays have global CSI, the DoF is upper bounded by
$\frac{MN}{M+N-1}$, regardless of the number of relays and the number
of antennas at the relays.
\end{proposition}
\begin{proof}
The $M\times N$ X channel without CSIT with relays can be upper
bounded by the $M\times N$ X channel with CSIT and relays. Note that
here we consider arbitrary number of relays with arbitrary number of
antennas. Reference
\cite{CadambeXNetwork} showed that with global CSI at all
nodes, the optimal DoF of the $M\times N$ X channel is
$\frac{MN}{M+N-1}$. Reference \cite{CadambeALLDOF} further showed that
relaying does not increase the DoF of X channels, when all nodes are
equipped with global CSI. This means that the $(M\times N)$-user X
channel with CSIT and relays with global CSI has optimal DoF $\frac{MN}{M+N-1}$, which
is clearly an upper bound for the $M\times N$ X channel without CSIT with relays.
\end{proof}
\begin{remark}
Note that since there is no assumption about whether the channel is
time varying or not in the arguments for outerbounds on DoF in references
\cite{CadambeXNetwork,CadambeALLDOF}, the DoF upper bound we have in
     {\it Proposition \ref{prop-upperbound-forX}} is valid for both
     time varying and constant channels.
\end{remark}

Now, we can proceed to construct the relay-aided interference alignment
schemes to show that, with the help of relays, the DoF upperbound
$\frac{MN}{M+N-1}$, which is obtained
by assuming global CSIT, is in fact
achievable without CSIT. Observe that for the $K$-user X channel, the DoF upperbound
reduces to $\frac{K^2}{2K-1}$.

\subsection{$M\times N$ X Channel with $J$ Relays with $L$ antennas}
We first consider the $M\times N$ X channel with
$J$ relays each having $L$ antennas and design transmission schemes
that can achieve the DoF upper bound in {\it Proposition
  \ref{prop-upperbound-forX}} without using CSIT.

\begin{theorem}\label{thm-MxN_X_JLrelay}
For the $M\times N$ X channel with $J$ relays each
having $L$ antennas, when the transmitters have no CSIT but the relays
have global CSI, a sufficient
condition to achieve the optimal DoF $\frac{MN}{M+N-1}$ is that
$J\ge \left\lceil\frac{(M-1)(N-1)}{L^2}\right\rceil$.
\end{theorem}
\begin{proof}
For the $M\times N$ X channel, each transmitter has a
message for each receiver, and we wish to deliver the $MN$
messages to the desired receivers in $M+N-1$ slots.

We label the relays
with $R_{i}$ where $i=1,2,\cdots,J$. For slots $t=1,2,\cdots,N$, the transmitters send the messages to the
receivers, and the relays remain silent. Specifically, the signal sent from transmitter $m$ at slot $t$ is
\begin{equation}
X_m(t)=d_{tm},
\end{equation} where $d_{tm}$ is the data stream carrying the message $W_{tm}$.

The received signals at receiver $n$ and relay $R_i$ are
\begin{equation}
Y_n(t)=\sum_{m=1}^Mh_{nm}(t)d_{tm}
\end{equation}
\begin{equation}
{\bf Y}_{R_{i}}(t)=\sum_{m=1}^M{\bf h}_{R_{i}m}(t)d_{tm}.
\end{equation} where ${\bf Y}_{R_{i}}(t)\in \mathbb{C}^L$.

For slots $t^\prime=N+1,\cdots,M+N-1$, each relay $R_i$ constructs a precoding matrix ${\rm \bf
  U}_{it}(t^\prime) \in \mathbb{C}^{L\times L}$ for the signals
received in each previous slot
$t$, and transmits the following
signal in slot $t^\prime$:
\begin{equation}
X_{R_i}(t^\prime)=\sum_{t=1}^N{\rm \bf U}_{it}(t^\prime){\bf Y}_{R_i}(t)
\end{equation}

In addition, for slot $t^\prime$, transmitter $1$ also sends the
following signal to the receivers:
\begin{equation}
X_1(t^\prime)=\sum_{n=1}^Nd_{n1}.
\end{equation}

The signal received at receiver $n$ for slot $t^\prime$ is thus
\begin{align}
Y_n(t^\prime)=&
\sum_{j=1}^N h_{n1}(t^\prime)d_{j1}\nonumber\\
&\qquad+\sum_{i=1}^J\sum_{t=1}^N\sum_{m=1}^M{\bf h}_{nR_i}(t^\prime)^T{\rm \bf U}_{it}(t^\prime){\bf h}_{R_{i}m}(t)d_{tm}.
\end{align}

After combining all the received signals from every slot, the
resulting signal can be expressed as in equation \eqref{eq:5} at the
beginning of next page.
\begin{figure*}[!t]
\hrulefill
\begin{align}
{\bf Y}_n &=\begin{array}{cc}
    \\
    n\\
    \\
    t^\prime\\
  \end{array}\left[
  \begin{array}{cc}
    {\bf 0}\\
    h_{n1}(n)\\
    {\bf 0}\\
    \left[h_{n1}(t^\prime)+\sum_{i=1}^{J}{\bf
      h}_{nR_i}(t^\prime)^T{\bf U}_{in}(t^\prime){\bf h}_{R_i1}(n)\right]^{(t^\prime)}
  \end{array}
  \right]d_{n1}+\sum_{k\neq 1}\left[
  \begin{array}{cc}
    {\bf 0}\\
    h_{nk}(n)\\
    {\bf 0}\\
    \left[\sum_{i=1}^{J}{\bf
      h}_{nR_i}(t^\prime)^T{\bf U}_{in}(t^\prime){\bf h}_{R_ik}(n)\right]^{(t^\prime)}
  \end{array}
  \right]d_{nk}\nonumber\\
&+\sum_{\gamma \neq n}\left(\begin{array}{cc}
    \\
    \gamma\\
    \\
    t^\prime\\
  \end{array}\left[
  \begin{array}{cc}
    {\bf 0}\\
    h_{n1}(\gamma)\\
    {\bf 0}\\
    \left[h_{n1}(t^\prime)+\sum_{i=1}^{J}{\bf
      h}_{nR_i}(t^\prime)^T{\bf U}_{i\gamma}(t^\prime){\bf h}_{R_i1}(\gamma)\right]^{(t^\prime)}
  \end{array}
  \right]d_{\gamma 1}\right.\nonumber\\
&\qquad \qquad \qquad \qquad \qquad \qquad \qquad \qquad \qquad \qquad
\left.+\sum_{k\neq 1}\begin{array}{cc}
    \\
    \gamma\\
    \\
    t^\prime\\
  \end{array}\left[
  \begin{array}{cc}
    {\bf 0}\\
    h_{nk}(\gamma)\\
    {\bf 0}\\
    \left[\sum_{i=1}^{J}{\bf
      h}_{nR_i}(t^\prime)^T{\bf U}_{i\gamma}(t^\prime){\bf h}_{R_ik}(\gamma)\right]^{(t^\prime)}
  \end{array}
  \right]d_{\gamma k}\right)\label{eq:5}
\end{align}
\hrulefill
\end{figure*} Note that in equation \eqref{eq:5}, $n$, $t^\prime$ and $\gamma$ outside the parenthesis
of the vectors denote the $n$th, $t^\prime$th and $\gamma$th entry of
the vectors, and we have utilized the notation defined in equations (\ref{eq:8})-(\ref{eq:10}).

In order to align all the interference messages into an $N-1$
dimensional space, we choose the precoding matrices ${\bf
  U}_{it}(t^\prime)$ such that
\begin{align}
&\frac{\sum_{i=1}^J{\bf h}_{nR_i}(t^\prime)^T{\bf
    U}_{i\gamma}(t^\prime){\bf
    h}_{R_ik}(\gamma)}{h_{nk}(\gamma)}\nonumber\\
&=\frac{h_{n1}(t^\prime)+\sum_{i=1}^J{\bf h}_{nR_i}(t^\prime)^T{\bf
    U}_{i\gamma}(t^\prime){\bf h}_{R_i1}(\gamma)}{h_{n1}(\gamma)},
\end{align} which can be written as
\begin{align}
&\sum_{i=1}^J{\bf h}_{nR_i}(t^\prime)^T{\bf
  U}_{i\gamma}(t^\prime)\left(h_{n1}(\gamma){\bf
  h}_{R_ik}(\gamma)-h_{nk}(\gamma){\bf
  h}_{R_i1}(\gamma)\right)\nonumber\\
&=h_{nk}(\gamma)h_{n1}(t^\prime), \label{eq:9}
\end{align} for all $t^\prime=N+1,\cdots, M+N-1$.

If we denote the entry for $p$th row and $q$th column of matrix ${\bf
    U}_{i\gamma}(t^\prime)$ as ${
    u}_{i\gamma,p,q}(t^\prime)$, where $p,q=1,\cdots,L$, we can define
a vector
\begin{equation}
{\bf u}(\gamma,t^\prime)=[{u}_{i\gamma,p,q}(t^\prime)]^{(ipq)},
\end{equation} where the notation
$[{u}_{i\gamma,p,q}(t^\prime)]^{(ipq)}$ is defined as in equation
\eqref{eq:8}-\eqref{eq:10}. We also define vectors
\begin{align}
&{\bf h}_{n,k}(\gamma,t^\prime)=\left[h_{nR_i,p}(t^\prime)\left(h_{n1}(\gamma)h_{R_ik,q}(\gamma)\right.\right.\nonumber\\
&\qquad \left. \left. -h_{nk}(\gamma)h_{R_i1,q}(\gamma)\right)\right]^{(ipq)},
\end{align} and matrix ${\bf H}(\gamma,t^\prime)$, which is
formed by taking ${\bf
  h}_{n,k}(\gamma,t^\prime)^T$ as its rows for all enumeration of $n$
and $k$.

All the linear equations can now be written as
\begin{equation}
{\bf H}(\gamma,t^\prime){\bf u}(\gamma,t^\prime)={\bf b}(\gamma,t^\prime),\label{eq:20}
\end{equation} where ${\bf
  b}(\gamma,t^\prime)=[h_{nk}(\gamma)h_{n1}(t^\prime)]^{(nk)}$.

Since we have one equation for each pair of $(n,k)$ where $k\neq 1,n \neq
\gamma$, there are $(M-1)(N-1)$ equations for each pair of fixed
$(\gamma,t^\prime)$.  On the other hand, each matrix ${\bf
  U}_{i\gamma}(t^\prime)$ can provide $L^2$ variables, which gives us $JL^2$ variables in
total. When $J\ge
\left\lceil\frac{(M-1)(N-1)}{L^2}\right\rceil$, we can guarantee that
there exist solutions to the equations to find the matrices ${\bf
  U}_{i\gamma}(t^\prime)$.

In the sequel, we drop the parameters $(\gamma,t^\prime)$ in the
expression for matrix ${\bf H}(\gamma,t^\prime)$ for
clarity. Since the channel coefficients are drawn from a continuous
distribution, the matrix ${\bf
  H}$ is of full rank almost surely. When the matrix ${\bf H}$ is square, the relays can find the precoding
matrix by calculating ${\bf  H}^{-1}{\bf b}$. When the matrix ${\bf
  H}$ is not square, since $J\ge
\left\lceil\frac{(M-1)(N-1)}{L^2}\right\rceil$, the vector ${\bf u}$ can be calculated using ${\bf
    H}^\dagger({\bf HH}^\dagger)^{-1}{\bf b}$. The calculation of the
precoding matrices for both cases only requires
global CSI at the relays and no cooperation between the relays is needed.

With the matrices ${\bf
  U}_{i\gamma}(t^\prime)$, all the interfering signals can be aligned
into an $N-1$ dimensional space. We now need to verify that the
interference and the signals carrying intended messages are linearly
independent. Since for receiver $n$, the signals carrying intended
messages and the interfering signals do not have non-zero entries in
the same row of the received signal vector from row 1 to row $N$, as shown in
equation (\ref{eq:5}), it is guaranteed that the signals are linearly
independent. As an example, consider the channel with $M=N=3$ and
receiver 1. The received signal is of the form

\begin{equation}
{\bf Y}_1=\left[
  \begin{array}{c}
    { a}\\
    { 0}\\
    { 0}\\
    { b}\\
    { c}
  \end{array}
  \right] d_{intended}+\left[
  \begin{array}{c}
    { 0}\\
    { e}\\
    { 0}\\
    { g}\\
    { h}
  \end{array}
  \right] d_{interference}. \label{eq:25}
\end{equation} It can be readily seen that the signal vectors carrying
intended messages and the ones with interference in equation
(\ref{eq:25}) are linearly independent.

This special structure of the received signals is
originated from the design of the transmission scheme. The fact that
the channel coefficients are drawn
from a continuous distribution guarantee that the desired data streams occupy the rest $M$
dimensional space, and thus a zero-forcing decoder can recover all the
desired messages to achieve DoF $\frac{MN}{M+N-1}$.
\end{proof}

\begin{remark}\label{remark:MN_X_Joint_Beamforming}
Note that in the scheme, when $JL^2=(M-1)(N-1)$, joint beamforming is
mandatory to obtain the precoding matrices ${\bf
  U}_{i\gamma}(t^\prime)$. This is because when $JL^2=(M-1)(N-1)$, the
matrix ${\bf H}(\gamma,t^\prime)$ in equation \eqref{eq:20} is invertible, and the vector ${\bf
  b}(\gamma,t^\prime)$ becomes zero if joint beamforming is not
utilized. The precoding matrices at the relays thus are all zero. In
this way, the interference can still be aligned since they occupy
different time indices than the intended signals, but there is not
sufficient dimension to decode the intended messages. Joint
beamforming, for this case, can guide the relays to steer the
interference such that they are aligned, and in the meantime guarantee
that there is sufficient signal dimension. On the other hand, when
$JL^2>(M-1)(N-1)$, joint beamforming is not required, since we can
always find a non-zero vector from the null space of matrix ${\bf H}(\gamma,t^\prime)$.
$\hfill{} \blacksquare$\end{remark}

\begin{remark}\label{remark-MN_X_JL_timevarying}
In the transmission scheme, the time varying nature of the channel is
crucial for the receivers to decode the intended signals. From
equation \eqref{eq:5}, we can see that when channel is not time
varying, the intended signals fall into a space of dimension 2, and
the receivers cannot decode all the intended messages.
$\hfill{} \blacksquare$\end{remark}

\begin{remark}\label{remark-1}
For the X channel without CSIT and without relays, reference
\cite{VazeNoCSIT} has shown that the DoF upperbound is 1 when the
channel experiences Rayleigh fading, provided that
all nodes are equipped with a single antenna. Our result shows that
relaying is useful to provide DoF gain for the X channel without
CSIT. This is to be contrasted with the result in reference
\cite{CadambeALLDOF}, which has
shown that relaying cannot provide DoF
gain for the X channel when global CSI is available at all the nodes.
$\hfill{} \blacksquare$\end{remark}

\begin{remark}
We have shown that using relays, we can achieve the optimal DoF for
the $M\times N$ X channel in finite channel uses. However, for the case with
global CSIT but no relays are available, the same optimal DoF is 
achievable using infinite channel uses, as shown in reference
\cite{CadambeXNetwork}.
$\hfill{} \blacksquare$\end{remark}

\begin{remark}\label{remark_X_MIMO}
The scheme we used in {\it Theorem \ref{thm-MxN_X_JLrelay}} can be
generalized to the case when each user has multiple antennas by
counting the number of equations required for interference alignment
and the number of variables that can be provided by the relays.
$\hfill{} \blacksquare$\end{remark}

We next investigate a special case of the general X channel, which is
the $K$-user X channel with a single relay equipped with multiple
antennas. For this case, we can design a different scheme using the
available spatial dimension at the relay, which can provide more insights
regarding how interference signals are aligned and has lower computational
complexity for the relay to obtain the precoding matrices.

\subsection{The $K$-user X Channel with one Multi-antenna Relay}\label{sec:k-user-x}
For the $K$-user X channel without CSIT, when we have a relay with $K$ antennas,
the relay can decode all the data streams sent from the transmitters,
for example with zero-forcing,
if each transmitter only sends a single data stream with DoF 1. Since
the relay has global CSI, clearly it can perform appropriate precoding
to align the interfering signals at the receivers. The result from
{\it Theorem \ref{thm-MxN_X_JLrelay}} implies that for this case,
$K-1$ antennas are in fact sufficient for the relay to align the
interference at the receivers to achieve the DoF upperbound
$\frac{K^2}{2K-1}$.

When the relay has multiple antennas, we can use a different strategy
than the one we used to prove {\it Theorem \ref{thm-MxN_X_JLrelay}}
to achieve the DoF upper bound. To better illustrate the transmission strategy, we first
provide an example for the 3-user X channel with a relay having 2
antennas, and then generalize the scheme to the $K$-user case. Note
that when $K=3$, the DoF upperbound becomes $\frac{9}{5}$.

\subsubsection{3-user X Channel with a Relay with 2 Antennas}
\begin{corollary}\label{prop-3x3-X-2antRelay}
For the $3$-user X channel without CSIT with relays, optimal DoF $\frac{9}{5}$ is achievable using a relay
with 2 antennas and global CSI.
\end{corollary}
\begin{proof}
We denote the data stream from transmitter $i$ to receiver $j$ as
$d_{ji}$, $i,j=1,2,3$, and data stream $d_{ji}$ carries a message $w_{ji}$. To achieve the DoF $\frac{9}{5}$, we let each
transmitter send one message to each receiver in 5 time slots. Note
that the channel is assumed to be time varying for each slot, and the channel coefficients
are drawn from a continuous distribution.

In the first 3 slots, the transmitters send messages to the
receivers, while the relay keeps silent. Specifically, in slot
$t$, all the 3 transmitters send the messages intended for receiver $t$:
\begin{equation}
X_{k}(t)=d_{tk}
\end{equation}
where $t,k=1,2,3.$ At slot $t$, the received signals at the receivers
and the relay are
\begin{align}
&Y_{1}(t)=h_{11}(t)d_{t1}+h_{12}(t)d_{t2}+h_{13}(t)d_{t3}\\
&Y_{2}(t)=h_{21}(t)d_{t1}+h_{22}(t)d_{t2}+h_{23}(t)d_{t3}\\
&Y_{3}(t)=h_{31}(t)d_{t1}+h_{32}(t)d_{t2}+h_{33}(t)d_{t3}\\
&{\bf Y}_R(t)={\bf h}_{R1}(t)d_{t1}+{\bf h}_{R2}(d)_{t2}+{\bf h}_{R3}(t)d_{t3},
\end{align} where we discarded the channel noise since we are
considering the DoF of the channel.

In the remaining 2 slots, the relay needs to provide each receiver
with two more equations such that the intended
messages, which are the unknown variables $d_{ji}$ in the equations,
can be recovered. In the meantime, all the interference data
streams must be kept in a 2-dimensional space at each receiver to
achieve the optimal DoF. Since the relay has 2
antennas, it cannot decode all the three messages from each user to perform appropriate precoding in the remaining 2
slots. However, as we shall see, the spatial dimensions available at
the relay can still be utilized to align the interference.

The relay first performs a linear transformation to the received
signals using vectors ${\bf u}_{i}(t)\in \mathbb{C}^2$, where
$i,t=1,2,3$, $i\neq t$. Specifically, for
$t=1$, we want to partially align the interference caused by the
messages intended for receiver 1. We
design the vectors ${\bf u}_{2}(1)$ and ${\bf u}_{3}(1)$ such that they satisfy
\begin{equation}
{\bf u}_2(1)^T{\bf h}_{R2}(1)=h_{22}(1)\qquad{\bf u}_2(1)^T{\bf h}_{R3}(1)=h_{23}(1),
\end{equation}
\begin{equation}
{\bf u}_3(1)^T{\bf h}_{R2}(1)=h_{32}(1)\qquad{\bf u}_3(1)^T{\bf h}_{R3}(1)=h_{33}(1).
\end{equation}

Since we have two variables with two equations for each vector ${\bf
  u}_i(1)$ and the channel
coefficients are drawn from a continuous distribution, we can
guarantee the existence of ${\bf u}_2(1)$ and ${\bf u}_3(1)$ almost
surely. We can then obtain the following signals by taking the inner
products between the vector ${\bf u}_2(1)$ $({\bf u}_3(1))$ and the received signal vector from slot 1:
\begin{align}
&{\bf u}_2(1)^T{\bf Y}_{R}(1)={\bf u}_2(1)^T{\bf h}_{R1}(1)d_{11}+h_{22}(1)d_{12}+h_{23}(1)d_{13}\label{eq:23}\\
&{\bf u}_3(1)^T{\bf Y}_{R}(1)={\bf u}_3(1)^T{\bf h}_{R1}(1)d_{11}+h_{32}(1)d_{12}+h_{33}(1)d_{13}\label{eq:24}.
\end{align}

These two signals are useful for receiver 1, since they
contain the messages that are intended for it. However, the messages
$d_{11}$, $d_{12}$ and $d_{13}$ are interference for receiver 2 and 3. Using
the linear transformation provided by vector ${\bf u}_2(1)$ or ${\bf
  u}_3(1)$, we can see that the channel coefficients for $d_{12}$ and
$d_{13}$ in equation \eqref{eq:23} are the same as the signal received at
receiver 2. Similarly, the channel coefficients for $d_{12}$ and
$d_{13}$ in equation \eqref{eq:24} are the same as the signal received
at receiver 3. If we can keep ${\bf
  u}_2(1)^T{\bf Y}_{R}(1)$ away from receiver 3, and keep ${\bf
  u}_3(1)^T{\bf Y}_{R}(1)$ away from receiver 2, part of
the interference is aligned at receiver 2 and receiver 3. This can be done
by sending ${\bf
  u}_2(1)^T{\bf Y}_{R}(1)$ and ${\bf
  u}_3(1)^T{\bf Y}_{R}(1)$ along the directions
\begin{equation}
{\bf v}_{12}(t) \perp {\bf h}_{3R}(t)\quad{\rm and}\quad{\bf v}_{13}(t) \perp {\bf h}_{2R}(t),
\end{equation} respectively, where $t=4,5$.

For the interference caused by the messages for receiver 2 and
receiver 3, we design the 
precoding vectors ${\bf u}_1(2)$, ${\bf u}_3(2)$, ${\bf u}_1(3)$, and ${\bf
  u}_2(3)$ in the same fashion as we design the vectors ${\bf
  u}_2(1)$, ${\bf u}_3(1)$, which have the following properties:
\begin{align}
&{\bf u}_1(2)^T{\bf Y}_{R}(2)={\bf u}_1(2)^T{\bf h}_{R2}(2)d_{22}+h_{11}(2)d_{21}+h_{13}(2)d_{23}\\
&{\bf u}_3(2)^T{\bf Y}_{R}(2)={\bf u}_3(2)^T{\bf h}_{R2}(2)d_{22}+h_{31}(2)d_{21}+h_{33}(2)d_{23}\\
&{\bf u}_1(3)^T{\bf Y}_{R}(3)={\bf u}_1(3)^T{\bf h}_{R3}(3)d_{33}+h_{11}(3)d_{31}+h_{12}(3)d_{32}\\
&{\bf u}_2(3)^T{\bf Y}_{R}(3)={\bf u}_2(3)^T{\bf h}_{R3}(3)d_{33}+h_{21}(3)d_{31}+h_{22}(3)d_{32}.
\end{align}

In order to transmit the signals along their intended directions, we
now define the following beamforming vectors, which is similar as the
vectors ${\bf v}_{12}(t)$ and ${\bf v}_{13}(t)$:
\begin{equation}
{\bf v}_{21}(t) \perp {\bf h}_{3R}(t)\quad{\bf v}_{23}(t) \perp {\bf h}_{1R}(t)
\end{equation}
\begin{equation}
{\bf v}_{31}(t) \perp {\bf h}_{2R}(t)\quad{\bf v}_{32}(t) \perp {\bf h}_{1R}(t)
\end{equation}
where $t=4,5$. We can choose the vectors such that they have unit
power, and satisfy
\begin{equation}
{\bf v}_{31}(t)={\bf v}_{13}(t)={\bf v}_{2}^{\perp}(t),
\end{equation}
\begin{equation}
{\bf v}_{12}(t)={\bf v}_{21}(t)={\bf v}_{3}^{\perp}(t),
\end{equation}
\begin{equation}
{\bf v}_{23}(t)={\bf v}_{32}(t)={\bf v}_{1}^{\perp}(t).
\end{equation}

Using the linear transformation and beamforming provided above, interference is only
partially aligned. To align the rest of the interference, we let the
relay choose a scaling factor $\alpha_{ij}(t)$ for each signal it
wishes to send to the receivers, and produce the signals
to be transmitted for slot 4 and slot 5 as shown in equation
\eqref{eq:333} at the beginning of next page, where the scalars $\alpha_{ij}(t)$ are to be determined later.
\begin{figure*}[!t]
\hrulefill
\begin{align}
{\bf X}_R(t) &= \alpha_{12}(t){\bf v}_{3}^{\perp}(t)\left({\bf
  u}_2(1)^T{\bf Y}_{R}(1)\right)+\alpha_{13}(t){\bf
  v}_{2}^\perp(t)\left({\bf u}_3(1)^T{\bf
  Y}_{R}(1)\right)+\alpha_{21}(t){\bf v}_{3}^\perp(t)\left({\bf
  u}_1(2)^T{\bf Y}_{R}(2)\right)\nonumber\\
&+\alpha_{23}(t){\bf v}_{1}^\perp(t)\left({\bf u}_3(2)^T{\bf
  Y}_{R}(2)\right)+\alpha_{31}(t){\bf v}_{2}^\perp(t)\left({\bf
  u}_1(3)^T{\bf Y}_{R}(3)\right)+\alpha_{32}(t){\bf
  v}_{1}^\perp(t)\left({\bf u}_2(3)^T{\bf Y}_{R}(3)\right)\label{eq:333}
\end{align} 
\hrulefill
\end{figure*}

For slots 4 and 5, the transmitters also send the following signals
to the receivers:
\begin{equation}
X_k(t)=d_{kk}
\end{equation} where $k=1,2,3.$ Note that other combinations of transmitted
messages also work for our scheme.

The received signals at the receivers can be expressed as
\begin{align}
&Y_1(t)=h_{11}(t)d_{11}+h_{12}(t)d_{22}+h_{13}(t)d_{33}+{\bf h}_{1R}(t)^T{\bf X}_R(t)\\
&Y_2(t)=h_{21}(t)d_{11}+h_{22}(t)d_{22}+h_{23}(t)d_{33}+{\bf h}_{2R}(t)^T{\bf X}_R(t)\\
&Y_3(t)=h_{31}(t)d_{11}+h_{32}(t)d_{22}+h_{33}(t)d_{33}+{\bf h}_{3R}(t)^T{\bf X}_R(t)
\end{align}

If we combine all the received signals from 5 slots into a vector in
$\mathbb{C}^5$, the resulting signal is shown in equation \eqref{eq:1} at
the beginning of next page.
\begin{figure*}[!t]
\begin{align}
&{\bf Y}_1 =\left[
  \begin{array}{cc}
    h_{11}(1) \\
    0\\
    0\\
    h_{11}(4)+\alpha_{12}(4)h_{1R}^{\perp 3}(4)\mu_2^{R1}(1)+\alpha_{13}(4)h_{1R}^{\perp 2}(4)\mu_3^{R1}(1) \\
    h_{11}(5)+\alpha_{12}(5)h_{1R}^{\perp 3}(5)\mu_2^{R1}(1)+\alpha_{13}(5)h_{1R}^{\perp 2}(5)\mu_3^{R1}(1) \\
  \end{array}
  \right]d_{11}+\left[
  \begin{array}{cc}
    h_{12}(1) \\
    0\\
    0\\
    \alpha_{12}(4)h_{1R}^{\perp 3}(4)h_{22}(1)+\alpha_{13}(4){h}_{1R}^{\perp 2}(4)h_{32}(1) \\
    \alpha_{12}(5)h_{1R}^{\perp 3}(5)h_{22}(1)+\alpha_{13}(5){h}_{1R}^{\perp 2}(5)h_{32}(1) \\
  \end{array}
  \right]d_{12}\nonumber\\
&+\left[
  \begin{array}{cc}
    h_{13}(1) \\
    0\\
    0\\
    \alpha_{12}(4)h_{1R}^{\perp 3}(4)h_{23}(1)+\alpha_{13}(4){h}_{1R}^{\perp 2}(4)h_{33}(1) \\
    \alpha_{12}(5)h_{1R}^{\perp 3}(5)h_{23}(1)+\alpha_{13}(5){h}_{1R}^{\perp 2}(5)h_{33}(1) \\
  \end{array}
  \right]d_{13}+\left[
  \begin{array}{cc}
    0\\
    h_{11}(2) \\
    0\\
    \alpha_{21}(4){h}_{1R}^{\perp 3}(4)h_{11}(2) \\
    \alpha_{21}(5){h}_{1R}^{\perp 3}(5)h_{11}(2) \\
  \end{array}
  \right]d_{21}+\left[
  \begin{array}{cc}
    0\\
    h_{13}(2) \\
    0\\
    \alpha_{21}(4){h}_{1R}^{\perp 3}(4)h_{13}(2) \\
    \alpha_{21}(5){h}_{1R}^{\perp 3}(5)h_{13}(2) \\
  \end{array}
  \right]d_{23} \nonumber\\
&+\left[
  \begin{array}{cc}
    0\\
    h_{12}(2) \\
    0\\
    h_{12}(4)+\alpha_{21}(4){h}_{1R}^{\perp 3}(4)\mu_1^{R2}(2) \\
    h_{12}(5)+\alpha_{21}(5){h}_{1R}^{\perp 3}(5)\mu_1^{R2}(2) \\
  \end{array}
  \right]d_{22}+\left[
  \begin{array}{cc}
    0\\
    0\\
    h_{11}(3) \\
    \alpha_{31}(4){h}_{1R}^{\perp 2}(4)h_{11}(3) \\
    \alpha_{31}(5){h}_{1R}^{\perp 2}(5)h_{11}(3) \\
  \end{array}
  \right]d_{31}\nonumber\\
&\qquad \qquad \qquad \qquad \qquad \qquad \qquad \qquad \qquad +\left[
  \begin{array}{cc}
    0\\
    0\\
    h_{12}(3) \\
    \alpha_{31}(4){h}_{1R}^{\perp 2}(4)h_{12}(3) \\
    \alpha_{31}(5){h}_{1R}^{\perp 2}(5)h_{12}(3) \\
  \end{array}
  \right]d_{32}+\left[
  \begin{array}{cc}
    0\\
    0\\
    h_{13}(3) \\
    h_{13}(4)+\alpha_{31}(4){h}_{1R}^{\perp 2}(4){\mu}_1^{R3}(3) \\
    h_{13}(5)+\alpha_{31}(5){h}_{1R}^{\perp 2}(5){\mu}_1^{R3}(3) \\
  \end{array}
  \right]d_{33} \label{eq:1}
\end{align}
\hrulefill
\end{figure*}
where we denote $h_{kR}^{\perp i}(t)={\bf h}_{kR}(t)^T{\bf v}_{i}^\perp(t)$, $\mu_k^{Ri}(t)={\bf u}_k(t)^T{\bf h}_{Ri}(t)$.

From the above expression, we can see that the data streams $d_{21}$
and $d_{23}$ are aligned in a one-dimensional space, and the data
streams $d_{31}$ and $d_{32}$ are aligned in a one dimensional
space. To align the data stream $d_{22}$ with $d_{21}$ and $d_{23}$, 
we choose
\begin{equation}
\frac{h_{12}(t)+\alpha_{21}(t){h}_{1R}^{\perp 3}(t)\mu_1^{R2}(2)}{h_{12}(2)}=\alpha_{21}(t){h}_{1R}^{\perp 3}(t),
\end{equation} which is equivalent as
\begin{equation}
\alpha_{21}(t)=\frac{h_{12}(t)}{(h_{12}(2)-\mu_1^{R2}(2))h_{1R}^{\perp 3}(t)},\label{eq:2}
\end{equation} where $t=4,5.$

Similarly, to align $d_{33}$ with $d_{31}$ and $d_{32}$ we choose
\begin{equation}
\alpha_{31}(t)=\frac{h_{13}(t)}{(h_{13}(3)-\mu_1^{R3}(3))h_{1R}^{\perp 2}(t)}.\label{eq:6}
\end{equation}

The remaining parameters
$\alpha_{12}(t),\alpha_{32}(t),\alpha_{13}(t),\alpha_{23}(t)$ can be determined in a similar fashion. It is
easy to verify that the data streams $d_{11}$, $d_{12}$ and $d_{13}$
still occupy a 3-dimension space with the specified parameters
$\alpha_{ij}(t)$. This argument holds at receiver 2 and receiver 3 as
well. Hence using the proposed scheme, we can transmit a total of 9
messages using 5 slots, which proves the achievability of DoF
$\frac{9}{5}$.
\end{proof}

\begin{remark}
We can see from equations \eqref{eq:2} and \eqref{eq:6} that joint
beamforming is a key step to achieve the DoF upper bound. This is
because without joint beamforming, i.e., transmitters stay silent for
slot 4 and slot 5, the channel coefficients $h_{ij}(4)$ and
$h_{ij}(5)$ are all zero. As a result, all the parameters $\alpha_{ij}(t)$
become zero. Similar as {\it Remark
  \ref{remark:MN_X_Joint_Beamforming}}, without joint beamforming, the
interference signals can still be aligned, but there is not sufficient
dimension for the receivers to decode the intended signals.
$\hfill{} \blacksquare$\end{remark}

\begin{remark}
From equation \eqref{eq:1}, we can see that the channel needs to be
time varying for the receivers to have sufficient dimension to decode
the intended signals, following similar arguments as in {\it Remark \ref{remark-MN_X_JL_timevarying}}.
$\hfill{} \blacksquare$\end{remark}

The idea of the above transmission strategy is to use the limited
spatial dimensions available at the relay to first {\it partially} align the
interference, and then align the rest of the interference through joint
beamforming with the transmitters. Without the relay, the transmitters
cannot send the signals at the intended directions for interference
alignment since there is no CSIT, and reference
\cite{VazeNoCSIT} has shown that the DoF of the X channel for this
case collapses to 1. The advantage of having the relays to assist
interference alignment 
is thus obvious. Using the ideas from the example for the 3-user X channel with a 2-antenna
relay, we can now generalize the result to the $K$-user case.

\subsubsection{$K$-user X Channel with one $(K-1)$-antenna Relay}
\begin{corollary}\label{theorem:k-user-x-multiant-relay}
For the $K$-user X channel without CSIT with relays, the optimal DoF
$\frac{K^2}{2K-1}$ is achievable using one relay with $K-1$ antennas and
global CSI.
\end{corollary}
\begin{proof}
The achievability of DoF $\frac{K^2}{2K-1}$ follows the idea from {\it Corollary
  \ref{prop-3x3-X-2antRelay}}, and the detailed scheme is provided in Appendix
\ref{sec:proof-theor-k-user-x-multiant-relay}.
\end{proof}

\begin{remark}The schemes
provided in {\it Corollary \ref{prop-3x3-X-2antRelay}} and {\it
  Corollary \ref{theorem:k-user-x-multiant-relay}} can be seen as
specific construction of the precoding matrices at the relay, where partial
interference alignment and joint beamforming are utilized. The scheme
we used in {\it Corollary \ref{prop-3x3-X-2antRelay}} and {\it
  Corollary \ref{theorem:k-user-x-multiant-relay}} has more of a
straight forward physical
interpretation, and more importantly, it has lower computational complexity since it
only requires $K\times K$ matrix inversion when
finding the vectors ${\bf u}_i(t)$ and ${\bf v}_{ti}(t^\prime)$. In comparison,
the scheme we used in {\it Theorem \ref{thm-MxN_X_JLrelay}} requires matrix inversion operation of matrices
with dimension $K^2\times K^2$.
$\hfill{} \blacksquare$\end{remark}

\begin{remark}
For the general $M\times N$ X channel with $L$-antenna relays, we can also design a
transmission scheme that first uses partial interference alignment to
align $L$ interfering signals, and then uses joint
beamforming to align the rest of the interfering signals. The scheme
can be designed using similar ideas as in the proof of {\it
  Corollary \ref{prop-3x3-X-2antRelay}} and {\it
  Corollary \ref{theorem:k-user-x-multiant-relay}}, and thus is omitted
here.
$\hfill{} \blacksquare$\end{remark}

\subsection{$K$-user X channel with $J$ Single Antenna Relays}
We now consider the $K$-user X channel with
multiple single-antenna relays. From {\it Theorem
  \ref{thm-MxN_X_JLrelay}}, the condition to achieve the same optimal
DoF as the case when CSIT is available is summarized in the following corollary.

\begin{corollary}\label{theorem-K-X-single}
For the $K$-user X channel with single antenna relays, when there is
no CSIT but global CSI is available at the relays, a sufficient condition to achieve the optimal DoF
$\frac{K^2}{2K-1}$ is to have $(K-1)^2$ relays.
\end{corollary}

\begin{remark}
{\it Corollary \ref{theorem-K-X-single}} showed that if there are not
enough number of antennas at the relays, we can use more relays to
compensate the lack of spatial dimensions. If we consider the total
number of antennas at all the relays, we can see that the lack
  of spatial dimensions at the relays increases the total number of
  antennas needed to achieve the optimal DoF from $K-1$ to $(K-1)^2$.
$\hfill{} \blacksquare$\end{remark}

\begin{remark}
For $K$-user X channel with single antenna relays, the number of
relays required to achieve the DoF upper bound is
$\mathcal{O}(K^2)$. This clearly places more overhead for the relays to
obtain the global CSI as compared
to obtaining global CSI at the $K$ transmitters. If we want to keep a
comparable overhead and employ $K$ relays only, we can only allow
$\lfloor\sqrt{K}+1 \rfloor$ users to transmit, which yields a DoF to the order of
$\mathcal{O}(\sqrt{K})$. 
$\hfill{} \blacksquare$\end{remark}

We have seen that for the X channel without CSIT, relaying can provide
DoF gain to achieve the optimal DoF. It is trivial to see that the
same is true for the setting when the transmitters have delayed CSIT,
since one can always ignore the delayed CSIT and employ the same
scheme. We next consider the case where the relays have delayed CSI.

\subsection{Full CSI vs Delayed CSI at the Relay}
In this section, we
investigate the DoF of the $K$-user X channel without CSIT with one
$(K-1)$-antenna relay under the
assumption that the relay has delayed CSI. We first
consider the $K$-user X channel with one $K$-antenna relay, which clearly
provides a DoF upperbound to the case with a $(K-1)$-antenna relay.

\begin{theorem}
For the $K$-user X channel with a $K$-antenna relay, when there is no
CSIT and only delayed CSI is available at the relay,
the DoF is given by
\begin{equation}
\frac{K}{1+\frac{1}{2}+\cdots+\frac{1}{K}}\label{eq:4}
\end{equation}
\end{theorem}
\begin{proof}
The achievability of this DoF can be obtained using a similar strategy
as in \cite{StaleCSIT}. The scheme in \cite{StaleCSIT} is designed for the
$K$-user broadcast channel and consists of $K$ phases, where in phase $1$, the transmitter sends the messages to the
receivers. In slot $t=1,\cdots,K$ for phase $1$, the transmitter
sends ${\bf X}(t)=\left(d_{t1},d_{t2},\cdots,d_{tK}\right)^T$, where $d_{ti}$ is the $i$th message intended for
receiver $t$. The transmission scheme used for this phase can be
implemented for the $K$-user X channel. Since the relay has $K$ antennas and
delayed CSI, it can decode all the messages, and then it can
act as the transmitter in the broadcast channel to implement the transmission
scheme for the rest of the phases to achieve the DoF specified by (\ref{eq:4}).

To upper bound the DoF of the channel, we combine all the transmitters
and the relay, which yields a
broadcast channel with $2K$ antennas at the transmitter with
delayed CSIT. The outerbounds in references
\cite{StaleCSIT} and \cite{VazeBCdelayed} can then be used to obtain
equation (\ref{eq:4}).
\end{proof}

Recall that for the $K$-user X channel without CSIT, when the relay has global CSI,
we can achieve the optimal DoF $\frac{K^2}{2K-1}$ with only $K-1$
antennas at the relay. For the case with delayed CSI at the relay, when the relay
has $K-1$ antennas, the DoF at most equals equation \eqref{eq:4}.
It is clear that for the
$K$-user X channel without CSIT, global CSI at the relay can
provide a DoF gain, compared to the case when only delayed
CSI is available at the relay.

\section{Relay-aided Interference Alignment for $K$-user Interference
  Channel}\label{sec:relay-aided-interf-IC}
In this section, we investigate impact of relays on the DoF of the
$K$-user interference channel without CSIT, letting the relays utilize the time/frequency/spatial
dimensions available to steer the signals into the desired
directions. The goal is once again to recover the optimal DoF with CSIT. Relays are assumed to have
global CSI. Following similar arguments as in {\it Proposition
  \ref{prop-upperbound-forX}}, we first propose a DoF upper bound for
this channel.

\begin{proposition}
The DoF for the $K$-user interference channel without CSIT but with
the presence of relays with global CSI is upper bounded by $\frac{K}{2}$.
\end{proposition}
\begin{proof}
The DoF for the $K$-user interference channel without CSIT with relays
can be upper bounded by the $K$-user interference channel with CSIT
and relays. Since relaying does not provide any DoF
gain for interference channel with global CSI at all nodes
\cite{CadambeALLDOF}, the optimal DoF for $K$-user interference channel with full
CSIT, which is shown in reference \cite{CadambeKuser} to be
$\frac{K}{2}$, can be an upper bound for the $K$-user interference channel without CSIT with relays.
\end{proof}

\subsection{$J$ relays with $L$ antennas}
We first consider the most general case for the $K$-user interference channel, where we have $J$
relays each equipped with $L$ antennas.
\begin{theorem}\label{thm_KIC_JLR}
For the relayd- aided $K$-user interference channel without CSIT, when there is global
CSI at the relays, the optimal DoF
$\frac{K}{2}$ can be achieved using
$\left\lceil\frac{K(K-2)}{L^2}\right\rceil$ relays with $L$ antennas.
\end{theorem}
\begin{proof}
To show the achievability of DoF $\frac{K}{2}$, we
construct a 2-slot transmission scheme.

In the first slot, each transmitter sends a message to the intended
receiver, i.e.,
\begin{equation}
X_k(1)=d_k,
\end{equation} where $d_k$ denotes the data stream carrying the
message $w_k$, and $k=1,\cdots,K$.

The signals received at receiver $k$ and relay $R_j$ are
\begin{equation}
Y_k(1)=\sum_{i=1}^Kh_{ki}(1)d_i,
\end{equation}
\begin{equation}
{\bf Y}_{R_j}(1)=\sum_{i=1}^K{\bf h}_{R_ji}(1)d_i,
\end{equation} where ${\bf Y}_{R_j}(1), {\bf h}_{R_ji}(1)\in
\mathbb{C}^{L}$. 

Since we use a 2-slot transmission scheme, the signal space at the
receivers has 2 dimensions in time. To decode the intended message,
the receivers need to keep all the other
$K-1$ interference signals aligned in a one dimensional space. To
this end, relay $R_j$ applies a
precoding matrix to the received signal vector, and transmits the following signal vector in the second slot:
\begin{equation}
{\bf X}_{R_j}(2)={\bf U}_j{\bf Y}_{R_j}(1)
\end{equation} where ${\bf U}_j \in \mathbb{C}^{L\times L}$,
which is to be determined later. In the second slot, we also let the receiver perform
joint beamforming to transmit
\begin{equation}
X_k(2)=d_k.
\end{equation}

The received signal at receiver $k$ for slot 2 can be expressed as
\begin{align}
Y_k(2)&=\sum_{i=1}^Kh_{ki}(2)d_i+\sum_{j=1}^J{\bf h}_{kR_j}^T(2){\bf
  x}_{R_j}\\
&=\sum_{i=1}^Kh_{ki}(2)d_i+\sum_{j=1}^J\sum_{i=1}^K{\bf h}_{kR_j}^T(2){\bf
  U}_{j}{\bf h}_{R_ji}(1)d_i.
\end{align}

Grouping the received signals at receiver $k$ from 2 slots into
vector form, we have
\begin{align}
&{\bf Y}_k=\left(
  \begin{array}{cc}
    h_{kk}(1) \\
    h_{kk}(2)+\sum_{j=1}^J{\bf h}_{kR_j}^T(2){\bf
  U}_{j}{\bf h}_{R_jk}(1) \\
  \end{array}
  \right)d_k\nonumber\\
&+\sum_{i\neq k}\left(
  \begin{array}{cc}
    h_{ki}(1) \\
    h_{ki}(2)+\sum_{j=1}^J{\bf h}_{kR_j}^T(2){\bf
  U}_{j}{\bf h}_{R_ji}(1) \\
  \end{array}
  \right)d_i\label{eq:7}.
\end{align}

In order to align all the interference signals into a one dimensional space, we need
\begin{align}
&\frac{h_{ki}(2)+\sum_{j=1}^J{\bf h}_{kR_j}^T(2){\bf
  U}_{j}{\bf h}_{R_ji}(1)}{h_{ki}(1)}\nonumber\\
&=\frac{h_{kl}(2)+\sum_{j=1}^J{\bf h}_{kR_j}^T(2){\bf
  U}_{j}{\bf h}_{R_jl}(1)}{h_{kl}(1)}\label{eq:3}
\end{align}
where $i=2$ if $k=1$, and $i=1$ if $k\neq 1$, for all $l
\neq k, l\neq i$. Equation \eqref{eq:3} can be equivalently written as 
\begin{align}
&\sum_{j=1}^J{\bf h}_{kR_j}^T(2){\bf
  U}_{j}\left(\frac{{\bf h}_{R_ji}(1)}{h_{ki}(1)}-\frac{{\bf
    h}_{R_jl}(1)}{h_{kl}(1)}\right)\nonumber\\
&=\left(\frac{h_{kl}(2)}{h_{kl}(1)}-\frac{h_{ki}(2)}{h_{ki}(1)}\right).\label{eq:17}.
\end{align}

If we denote the entries of ${\bf U}_j$ as $u_{j,mn}$, where
$m,n=1,\cdots,L$, entries of ${\bf h}_{kR_j}(2)$ as
${h}_{kR_j,m}(2)$, and entries of ${\bf h}_{R_ji}(1)$ as
${h}_{R_ji,n}(1)$, then equation \eqref{eq:17} can be written as 
\begin{align}
&\sum_{j=1}^J\sum_{m=1}^L\sum_{n=1}^L{h}_{kR_j,m}(2)\left(\frac{{h}_{R_ji,n}(1)}{h_{ki}(1)}-\frac{{
    h}_{R_jl,n}(1)}{h_{kl}(1)}\right){
  u}_{j,mn}\nonumber\\
&=\left(\frac{h_{kl}(2)}{h_{kl}(1)}-\frac{h_{ki}(2)}{h_{ki}(1)}\right).
\end{align}

If we let
\begin{equation}
{\bf h}_{kl}=\left[{h}_{kR_j,m}(2)\left(\frac{{h}_{R_ji,n}(1)}{h_{ki}(1)}-\frac{{
    h}_{R_jl,n}(1)}{h_{kl}(1)}\right)\right]^{(jmn)},
\end{equation}
\begin{equation}
{\bf
  b}=\left[\frac{h_{kl}(2)}{h_{kl}(1)}-\frac{h_{ki}(2)}{h_{ki}(1)}\right]^{(kl)} \label{eq:18}
\end{equation}
 and reorganize $u_{j,mn}$ to form a vector
\begin{equation}
{\bf u}=\left[u_{j,mn}\right]^{(jmn)},
\end{equation} then all the linear equations can be written as
\begin{equation}
{\bf Hu}={\bf b},\label{eq:19}
\end{equation} where ${\bf H}$ is obtained by using ${\bf h}_{kl}^T$ 
as its rows for all the enumeration of $k$ and $l$, corresponding to
the order of indices $k$ and $l$ in ${\bf b}$. 

The matrix ${\bf H}$ has dimension $K(K-2)\times JL^2$, and it is full
rank almost surely since the entries of channel matrices are drawn
from a continuous distribution. In order to guarantee that the
interference is aligned, we need to have $JL^2\ge K(K-2)$ such that
we can find precoding matrices ${
\bf U}_j$ at the relays. When $JL^2\ge K(K-2)$, matrices ${\bf U}_j$
can be obtained from the null space of matrix ${\bf H}$ or inverting
the matrix ${\bf H}$. 

Now we need to show that the interference and the signal carrying
intended messages are linearly independent. We first observe that when
$JL^2\ge K(K-2)$, ${\bf u}={\bf H}^\dagger({\bf H}{\bf H}^\dagger)^{-1}{\bf
  b}$ is always a solution. The
matrices ${\bf U}_j$ are thus only linear functions of the channel
coefficients except for $h_{kk}(1)$ and $h_{kk}(2)$. From equation
\eqref{eq:7}, since interference is aligned, we have 
\begin{equation}
\lambda h_{ki}(1)=h_{ki}(2)+\sum_{j=1}^J{\bf h}_{kR_j}^T(2){\bf
  U}_{j}{\bf h}_{R_ji}(1)
\end{equation} for some $\lambda$. If the signal carrying intended
messages and the interference are also aligned, we must have 
\begin{equation}
\lambda h_{kk}(1)=h_{kk}(2)+\sum_{j=1}^J{\bf h}_{kR_j}^T(2){\bf
  U}_{j}{\bf h}_{R_jk}(1).
\end{equation} Since $\sum_{j=1}^J{\bf h}_{kR_j}^T(2){\bf
  U}_{j}{\bf h}_{R_jk}(1)$ is a linear function of channel
coefficients except for $h_{kk}(1)$ and $h_{kk}(2)$, the probability
that the signal carrying intended messages and the interference are
also aligned is zero, since the channel matrices are
generated from a continuous distribution. Therefore the receivers can
decode the intended messages using zero-forcing, and the DoF $K \over 2$
can be achieved almost surely.
\end{proof}
\begin{remark}\label{remark_IC_jointBF_channel}
In the above scheme, when we have $JL^2=K(K-2)$, the matrix ${\bf H}$
in equation \eqref{eq:19} is invertible. For this case, we must use
joint beamforming and the channel need to be time varying in order to
obtain non-zero precoding matrices ${\bf U}_j$ at the relays. This is
because when we do not use joint beamforming or the channel being
constant, the vector ${\bf b}$ on right hand side of equation
\eqref{eq:19} becomes zero, which results in all-zero precoding
matrices at the relay. This reduces the
available dimensions of the signal space at the receivers to
one, similar as the observation we have for the X channel in {\it
  Remark \ref{remark:MN_X_Joint_Beamforming}}. For this case, the intended signal and the interfering signals
are aligned together. By remaining silent, the relays are still
able to keep all the interference aligned. However, we need another dimension in the signal
space to separate the intended signal from the interference. Joint
beamforming and time varying channel, for this
case, allow the relays to facilitate interference alignment without
reducing the dimensions of the signal spaces at the receivers.
On the other hand, when we have $JL^2>K(K-2)$,
we can always find a non-zero vector ${\bf u}$ from the null space of
${\bf H}$, and thus the channel does not need to be time varying and
we do not have to use joint beamforming.
$\hfill{} \blacksquare$\end{remark}

\begin{remark}
In reference \cite{CadambeKuser}, the DoF $\frac{K}{2}$ for the $K$-user interference channel is achieved via channel
extension, which requires infinite channel uses to achieve exactly
$\frac{K}{2}$ degrees of freedom. In our scheme, however, the DoF is achieved via a two-slot transmission scheme.
$\hfill{} \blacksquare$\end{remark}

\begin{remark}
If we
assume that the channel coefficients are drawn from the Rayleigh
distribution, then it is shown in \cite{VazeNoCSIT} that the DoF for
the $K$-user interference channel without CSIT is upper bounded by
1. It is thus clear that relays can provide DoF gain for the $K$-user
interference channel without CSIT.
$\hfill{} \blacksquare$\end{remark}

\begin{remark}
The scheme we used for {\it Theorem \ref{thm_KIC_JLR}} can also be
applied to the case when the transmitters and the receivers have
multiple antennas.
$\hfill{} \blacksquare$\end{remark}

We next consider two special cases of the channel, namely the case
when there is a single relay with multiple antennas and the case when
there are multiple relays with a single antenna. 

\subsection{Single relay with multiple antennas} For this case, it
is easy to see that when a relay has $K$ antennas, the DoF upper bound
$\frac{K}{2}$ can be achieved using a 2-slot transmission scheme:
In the first slot, the transmitters send messages to the relay, and relay
decodes all messages. In the second slot, the relay broadcasts all the
messages to the receivers. The $K$ antennas at the relay can
provide sufficient spatial dimensions for the relay to decode and
broadcast the messages. However, from {\it Theorem \ref{thm_KIC_JLR}},
a sufficient condition to achieve the DoF $\frac{K}{2}$ is to have a
relay with $K-1$ antennas, which is summarized in the following corollary:

\begin{corollary}\label{theorem_IC_relay_K-1}
For the relay-aided $K$-user interference channel without CSIT, a
sufficient condition to achieve the optimal DoF $\frac{K}{2}$ is to
have $K-1$ antennas at the relay.
\end{corollary}

This result can be obtained as a special case from {\it Theorem
  \ref{thm_KIC_JLR}}. Note that this result was also obtained in \cite{Chen_Globecom} using
similar ideas. In fact, for this case, it is shown in
\cite{Chen_Globecom} that the $K-1$ antennas at the relay
is also a necessary condition to achieve the optimal DoF using linear
precoding schemes at the relay. 

From {\it Remark
  \ref{remark_IC_jointBF_channel}}, we observe two important features
for the case with a single relay equipped with $K-1$ antennas: The
channel does not need to be time
varying and there is no need for joint beamforming between
transmitters and the relay for the transmission in the second
slot.

\subsection{Multiple relays with single antenna}
We now focus on the case when relays only have a single antenna, and
investigate how many relays are needed to achieve the DoF 
$\frac{K}{2}$. From {\it Theorem
  \ref{thm_KIC_JLR}}, we have the following corollary.

\begin{corollary}\label{sec:relay-aided-interf}
For relay-aided $K$-user interference channel without CSIT, using the presence
of single antenna relays with global CSI, a sufficient
condition to achieve the optimal DoF $\frac{K}{2}$ is to have $K(K-2)$ relays.
\end{corollary}

Different from {\it Corollary \ref{theorem_IC_relay_K-1}}, for the case
when we have $K(K-2)$ relays with single antenna, joint beamforming between the
transmitters and the relays and the channel being time varying are two
important conditions to achieve the optimal DoF, as observed from {\it
Remark \ref{remark_IC_jointBF_channel}}.

\begin{remark}
The above scheme requires the number of single-antenna relays to be of the order
$\mathcal{O}(K^2)$, to achieve the optimal DoF for
the $K$-user interference channel with relays. It is then interesting to see how
much DoF we can achieve if the number of relays is of
the order $\mathcal{O}(K)$. For this case, we can consider a subset of
$\left\lfloor\sqrt{K}\right\rfloor$ transmitter-receiver pairs as a
$\left\lfloor\sqrt{K}\right\rfloor$-user interference channel. The achievable DoF is then
$\frac{\left\lfloor\sqrt{K}\right\rfloor}{2}$, which is still a significant
improvement compared to the DoF of the $K$-user interference channel with no relays under Rayleigh fading \cite{VazeNoCSIT}.
$\hfill{} \blacksquare$\end{remark}

\begin{remark}
For the $K$-user interference channel with relays, we
can also design a two-hop transmission scheme. However, this requires
more relays in general. Reference
\cite{RankovRelay} considered a two-hop interference network with
single antenna relays, and
showed that to achieve interference-free transmission, which implies
achieving DoF $K\over 2$, we need $K(K-1)+1$
relays. This is more than $K(K-2)$ relays that are needed for our
scheme. This is because in our scheme, there are more
dimension in the signal space that we can utilize due to the fully
connected nature of the channel and the interaction between
transmitters and the relays in the transmission in slot 2.
$\hfill{} \blacksquare$\end{remark}

\begin{figure*}[!t]
\hrulefill
\begin{align}
&Y_m(t^\prime)=\sum_{k=1}^Kh_{mk}(t^\prime)d_{kk}+{\bf h}_{mR}(t^\prime)^T{\bf X}_R(t^\prime)\\
&=\sum_{k=1}^Kh_{mk}(t^\prime)d_{kk}+{\bf
  h}_{mR}(t^\prime)^T\sum_{t\neq m,i= m}\alpha_{tm}(t^\prime){\bf v}_{tm}(t^\prime)\left({\bf u}_m(t)^T{\bf
  h}_{Rt}(t)d_{tt}+\sum_{n\neq t}h_{mn}(t)d_{tn}\right)+{\bf
  h}_{mR}(t^\prime)^T\cdot \nonumber\\
&\quad \sum_{t=m,i\neq m}\alpha_{mi}(t^\prime){\bf
  v}_{mi}(t^\prime)\left({\bf u}_i(m)^T{\bf
  h}_{Rm}(m)d_{mm}+\sum_{n\neq m}h_{in}(m)d_{mn}\right)\label{eq:262}
\end{align}
\hrulefill
\begin{align}
{\bf Y}_m &=\begin{array}{cc}
    \\
    m\\
    \\
    t^\prime\\
  \end{array}\left[
  \begin{array}{cc}
    {\bf 0}\\
    h_{mm}(m)\\
    {\bf 0}\\
    \left[h_{mm}(t^\prime)+\sum_{i\neq m}\alpha_{mi}(t^\prime){\bf
        h}_{mR}(t^\prime)^T{\bf v}_{mi}(t^\prime){\bf u}_i(m)^T{\bf
        h}_{Rm}(m)\right]^{(t^\prime)}
  \end{array}
  \right]d_{mm}\nonumber\\
&+\sum_{k\neq m}\begin{array}{cc}
    \\
    m\\
    \\
    t^\prime\\
  \end{array}\left[
  \begin{array}{cc}
    {\bf 0}\\
    h_{mk}(m)\\
    {\bf 0}\\
    \left[\sum_{i\neq m}\alpha_{mi}(t^\prime){\bf h}_{mR}(t^\prime)^T{\bf v}_{mi}(t^\prime)h_{ik}(m)\right]^{(t^\prime)}
  \end{array}
  \right]d_{mk}\nonumber\\
&+\sum_{\gamma \neq m}\left(\begin{array}{cc}
    \\
    \gamma\\
    \\
    t^\prime\\
  \end{array}\left[
  \begin{array}{cc}
    {\bf 0}\\
    h_{m\gamma}(\gamma)\\
    {\bf 0}\\
    \left[h_{m\gamma}(t^\prime)+\alpha_{\gamma m}(t^\prime){\bf
      h}_{mR}(t^\prime)^T{\bf v}_{\gamma m}(t^\prime){\bf
      u}_m(\gamma)^T{\bf h}_{R\gamma}(\gamma)\right]^{(t^\prime)}
  \end{array}
  \right]d_{\gamma \gamma}\right.\nonumber\\
&+\left.\sum_{k\neq \gamma}\begin{array}{cc}
    \\
    \gamma\\
    \\
    t^\prime\\
  \end{array}\left[
  \begin{array}{cc}
    {\bf 0}\\
    h_{mk}(\gamma)\\
    {\bf 0}\\
    \left[\alpha_{\gamma m}(t^\prime){\bf h}_{mR}(t^\prime)^T{\bf v}_{\gamma m}(t^\prime)h_{mk}(\gamma)\right]^{(t^\prime)}
  \end{array}
  \right]d_{\gamma k}\right)\label{eq:272}
\end{align}
\hrulefill
\begin{equation}
\alpha_{\gamma
  m}(t^\prime)=\frac{h_{m\gamma}(t^\prime)}{h_{m\gamma}(\gamma){\bf
    h}_{mR}(t^\prime)^T{\bf v}_{\gamma m}(t^\prime)h_{mk}(\gamma)-{\bf
    h}_{mR}(t^\prime)^T{\bf v}_{\gamma m}(t^\prime){\bf
    u}_m(\gamma)^T{\bf h}_{R\gamma}(\gamma)} \label{eq:263}
\end{equation} 
\hrulefill
\end{figure*}

\section{Conclusion}\label{sec:conclusion}
In this paper, we have investigated relay-aided interference alignment
schemes for the X channel and the interference channel, when no
channel state information (CSI) at the transmitters (CSIT) is
available. In particular, we have considered models where intermediate
relay nodes have access to CSI, and can compensate
for the lack of CSI at the transmitters. We have first investigated the
$M\times N$ X channel without CSIT assisted by relays with global
CSI. We have designed a transmission scheme and established sufficient conditions
between the number of relays and the
number of antennas at the relays such that the same optimal DoF as
the case when CSIT is available can be achieved. For the $K$-user interference
channel without CSIT, we have shown that relays can provide
interference alignment to achieve the optimal DoF $\frac{K}{2}$ using
a 2-slot transmission scheme.  In general, we have shown that the optimal DoF
$\frac{K}{2}$ can be achieved using $\left\lceil\frac{K(K-2)}{L^2}\right\rceil$
relays with $L$ antennas. 

In this paper, the focus has been on recovering the optimal DoF using
relays with global CSI, as if
transmitters had global CSI when in reality they have none. An interesting direction is quantifying
the impact of partial or delayed CSI at the relays on the DoF in the
presence of delayed or zero CSI at the transmitters. This is left as
future work.

\appendices
\section{Proof of Theorem \ref{theorem:k-user-x-multiant-relay}}\label{sec:proof-theor-k-user-x-multiant-relay}
We denote the message from transmitter $i$ to receiver $j$ as $d_{ji}$. We wish to send $K^2$ messages in $2K-1$ channel uses. In the first
 $K$ slots, the transmitters send the messages to the relay and the receivers, and in the rest $K-1$ slots, the relay performs partial interference alignment and joint beamforming with the transmitters to align all the interference into a $K-1$ dimensional space.

For slot $t=1,\cdots,K$, transmitter $k$ sends
\begin{equation}
X_k(t)=d_{tk}
\end{equation}
The signal received at receiver $m\in \{1,2,\cdots,K\}$ for slot $t$ is
\begin{equation}
Y_m(t)=\sum_{k=1}^Kh_{mk}(t)d_{tk}
\end{equation}
\begin{equation}
{\bf Y}_R(t)=\sum_{k=1}^K{\bf h}_{Rk}(t)d_{tk}
\end{equation}
where ${\bf Y}_R(t)\in\mathbb{C}^{K-1}$.

Now we need to obtain the vectors ${\bf u}_i(t)\in\mathbb{C}^{K-1}$ to
partially align the interference at the receivers. We let
\begin{equation}
{\bf u}_i(t)^T{\bf h}_{Rk}(t)=h_{ik}(t)
\end{equation}
where $k\neq t, i\neq t$. Since we have exactly $K-1$ equations to
solve for $K-1$ variables, and the channel coefficients are drawn from
the continuous distribution, there exist non-zero vectors ${\bf
  u}_i(t)$ almost surely.

We then have
\begin{equation}
{\bf X}_R^i(t)={\bf u}_i(t)^T{\bf Y}_R(t)={\bf u}_i(t)^T{\bf h}_{Rt}(t)d_{tt}+\sum_{k\neq t}h_{ik}(t).
\end{equation}

For each $X_R^i(t)$, we choose a weighting coefficient $\alpha_{ti}(t^\prime)$, where $t^\prime=K+1,\cdots,2K-1$, and a beamforming vector
${\bf v}_{ti}(t^\prime)$. We choose the beamforming vectors such that
${\bf v}_{ti}(t^\prime)\in \mathcal{N}([\left\{{\bf
    h}_{lR}(t^\prime)\right\}]^T)$, where $[\left\{{\bf
    h}_{lR}(t^\prime)\right\}]$ denotes a matrix taking the vector ${\bf
h}_{lR}(t^\prime)$ as its columns for all $l\neq t, l\neq i$. The matrix $[\left\{{\bf h}_{lR}(t^\prime)\right\}]^T$ has dimension $(K-2)\times (K-1)$, and
thus its null space is non-empty, which guarantees the existence of ${\bf v}_{ti}(t^\prime)$.

For slots $t^\prime = K+1, K+2, \cdots, 2K-1$, the relay transmits
\begin{equation}
{\bf X}_R(t^\prime)=\sum_{t=1}^K\sum_{i\neq t}\alpha_{ti}(t^\prime){\bf v}_{ti}(t^\prime){\bf u}_{i}(t)^T{\bf Y}_R(t).
\end{equation}

In the meantime, the transmitters send
\begin{equation}
X_k(t^\prime)= d_{kk}.
\end{equation}

The received signal at receiver $m$ for slot $t^\prime$ is
$Y_m(t^\prime)$ as shown in equation \eqref{eq:262} at the beginning of
this page.

We then combine the received signals from $2K-1$ slots into one vector
${\bf Y}_m$ as shown in equation \eqref{eq:272} at the beginning of this page.

For receiver $m$, $d_{mk}$, $k=1,\cdots,K$, are the messages that it needs to
decode, which should span a $K$ dimensional space. There are a total
of $2K-1$ dimensions available for the received signals, and hence we
should align the rest interference signals into a $K-1$ dimensional
space. With the help of the relay, we have already aligned the
interfering data streams
$d_{\gamma k}, \forall k \neq \gamma$, into a one dimensional space
for each fixed $\gamma$. If we can steer the data stream $d_{\gamma
  \gamma}$ into the same dimension of the signal space, then we are
able to keep all the interference into a $K-1$ dimensional space. This
is feasible by choosing the parameters $\alpha_{\gamma m}(t^\prime)$
such that equation \eqref{eq:263} at the beginning of previous page is
satisfied for all $t^\prime= K+1,\cdots,2K-1$.

It is easy to verify that after aligning the interference, the
intended messages $d_{mk}$ occupy a $K$-dimensional space, which does
not intersect with the $(K-1)$-dimensional space of the
interfering signals, and thus they can be decoded using a
zero-forcing decoder to completely eliminate the interference. Therefore we are able to send $K^2$ messages with $2K-1$
slots, and the DoF $\frac{K^2}{2K-1}$ is achievable.

\bibliographystyle{unsrt}

\end{document}